\documentclass[11pt]{amsart}
\usepackage[latin1]{inputenc} 
\usepackage{amsthm}
\usepackage{amsfonts}
\usepackage{verbatim}
\usepackage{graphicx}
\usepackage[italian, english]{babel}
\usepackage{amsfonts}
\usepackage{hyperref}
\usepackage{mathdots}

\newtheorem{defin}{Definition}[section]
\newtheorem{Def}[defin]{Definition}
\newtheorem{Lemma}[defin]{Lemma}
\newtheorem{Thm}[defin]{Theorem}

\newtheorem{Cor}[defin]{Corollary}
\newtheorem{Rem}[defin]{Remark}

\newcommand{\qbinom}[2]{\genfrac{[}{]}{0pt}{}{#1}{#2}}
\newcommand{\F}{\mathbb{F}}
\newcommand{\R}{\mathbb{R}}

\title[Running heading with forty characters or less]
      {Sample \LaTeX\ File for AMC}

\title[Bounds for projective codes]{Bounds for projective codes from semidefinite programming}

\author[Christine Bachoc, Alberto Passuello, Frank Vallentin]{}


\thanks{The third author was supported by Vidi grant 639.032.917 from the Netherlands Organization for Scientific Research (NWO)}

\begin{document}
\maketitle

\centerline{\scshape Christine Bachoc and Alberto Passuello}
\medskip
{\footnotesize
 \centerline{Univ. Bordeaux}
  \centerline{Institut de Math\'ematiques}
   \centerline{351, cours de la Lib\'eration}
   \centerline{F-33400 Talence, France}
   }

\medskip

\centerline{\scshape Frank Vallentin}
\medskip
{\footnotesize
 \centerline{Mathematisches Institut}
  \centerline{Universit\"at zu K\"oln}
   \centerline{Weyertal 86-90, 50931 K\"oln, Germany}
   }

\begin{abstract}
We apply the semidefinite programming method to derive bounds for
projective codes over a finite field.
\end{abstract}

\section{Introduction}  

In network coding theory, as introduced in \cite{Ahlswede},
information is transmitted through a directed graph. In general this
graph has several sources, several receivers, and a certain number of
intermediate nodes. Information is modeled as vectors of fixed length
over a finite field $\mathbb{F}_q$, called \emph{packets}. To improve
the performance of the communication, intermediate nodes should
forward random linear $\mathbb{F}_q$-combinations of the packets they
receive. This is the approach taken in the non-coherent communication
case, that is, when the structure of the network is not known a priori
\cite{benefits}. Hence, the \emph{vector space} spanned by the packets
injected at the source is globally preserved in the network when no
error occurs. This observation led Koetter and Kschischang
\cite{Koetter-Kschischang} to model network codes as subsets of
projective space $\mathcal{P}(\mathbb{F}_q^n)$, the set of linear
subspaces of $\mathbb{F}_q^n$, or of Grassmann space
$\mathcal{G}_q(n,k)$, the subset of those subspaces of
$\mathbb{F}_q^n$ having dimension~$k$.  Subsets of
$\mathcal{P}(\mathbb{F}_q^n)$ are called \emph{projective codes} while
subsets of the Grassmann space will be referred to as
\emph{constant-dimension codes} or \emph{Grassmann codes}.

As usual in coding theory, in order to protect the system from  errors, it is
desirable to select the elements of the code so that they are pairwise
as far as possible with respect to a suitable distance. The
\emph{subspace distance} between~$U$ and~$V$
$$d_S(U,V)=\dim (U+V) -\dim (U\cap V)=\dim U +\dim V -2\dim (U\cap
V)$$
was introduced in \cite{Koetter-Kschischang} for this purpose. It is
natural to ask how large a code with a given minimal distance can
be. Formally, we define
\begin{equation*}
\begin{cases}
A_q(n,d):=\ \max\{|\mathcal{C}|\ :\ \mathcal{C}\subset \mathcal{P}(\mathbb{F}_q^n),\ d_S(\mathcal{C})\geq d\}\\
A_q(n,k, 2\delta):=\ \max\{|\mathcal{C}|\ :\ \mathcal{C}\subset \mathcal{G}_q(n,k),\ d_S(\mathcal{C})\geq 2\delta\}
\end{cases}
\end{equation*}
where $d_S(\mathcal C)$ denotes the minimal subspace distance
among distinct elements of a code $\mathcal C$.
In this paper we will discuss and prove upper bounds for $A_q(n,d)$
and $A_q(n,k, 2\delta)$.

\subsection{Bounds for $A_q(n,k,2\delta)$}

Grassmann space $\mathcal{G}_q(n,k)$ is a homogeneous space under the
action of the linear group $GL_n(\F_q)$. Moreover, the group acts
\emph{distance transitively} when we use the subspace distance; the
orbits of $GL_n(\F_q)$ acting on pairs~$(U,V)$ of $\mathcal{G}_q(n,k)$
are characterized by the subspace distance~$d(U,V)$.  In other words,
Grassmann space is \emph{two-point homogeneous} under this action.

Due to this property, codes and designs in~$\mathcal{G}_q(n,k)$ can be
analyzed in the framework of Delsarte's theory, in the same way as
other classical spaces in coding theory, such as Hamming space and
binary Johnson space. In fact, $\mathcal{G}_q(n,k)$ is a
\emph{$q$-analog} of binary Johnson space;
see~\cite{Delsarte-Hahn}. The linear group plays the role of the
symmetric group for the Johnson space, while the dimension replaces
the weight function.

The classical bounds (anticode, Hamming, Johnson, Singleton) have been
derived for the Grassmann codes \cite{Koetter-Kschischang, WXSN,
  Xia-Fu}. The more sophisticated Delsarte linear programming bound
was obtained in \cite{Delsarte-Hahn}. However, numerical computations
indicate that it is not better than the anticode bound. Moreover, the
Singleton and anticode bounds have the same asymptotic behavior which
is attained by a family of Reed-Solomon-like codes constructed in
\cite{Koetter-Kschischang} and closely related to the rank-metric
Gabidulin codes.

\subsection{Bounds for $A_q(n,d)$}

In contrast to $\mathcal{G}_q(n,k)$, the projective space has a much
nastier behavior, essentially because it is not two-point
homogeneous. In fact it is not even homogeneous under the action of a
group. For example, the size of balls in this space depends not only
on their radius, but also on the dimension of their
center. Consequently, bounds for projective codes are much harder to
obtain.  Etzion and Vardy in \cite{Etzion-Vardy} provide a bound in
the form of the optimal value of a linear program, which is derived by
elementary reasoning involving packing issues. Up to now the Etzion-Vardy bound
is the only successful generalization of the classical bounds to
projective space.

\smallskip

\emph{In this paper we derive semidefinite programming bounds for
projective codes and compare them with the above mentioned bounds.}

\smallskip

In convex optimization, semidefinite programs generalize linear
programs and one can solve them by efficient algorithms \cite{Todd},
\cite{SDP}. They have numerous applications in combinatorial
optimization. The earliest is due to
Lov\'asz~\cite{Lovasz-Shannoncapacity} who found a semidefinite
programming upper bound, the theta number, for the independence number
of a graph.

Because a code with given minimal distance can be viewed as an
independent set in a certain graph, the theta number also applies to
coding theory. However, because the underlying graph is built on the
space under consideration, its size grows exponentially with the
parameters of the codes. So by itself the theta number is not an
appropriate tool, unless the symmetries of the space are taken into
account. A general framework for symmetry reduction techniques of
semidefinite programs is provided in \cite{inv-sdp}.  For the
classical spaces of coding theory, after symmetry reduction, the theta
number turns out to be essentially equal to the celebrated Delsarte
linear programming bound. For projective spaces, the symmetry reduction
was announced in \cite{Bachoc-Vallentin2} (see also
\cite{Bachoc-ITW}).  The program remains a semidefinite program (it
does not collapse to a linear program) but fortunately it has
polynomial size in the dimension $n$.

The relationship between Delsarte's linear programming bound and the
theta number was recognized long ago in \cite{Schrijver-comparaison}
and \cite{McEliece-Lovasz}. Recently, more applications of
semidefinite programming to coding theory have been developed, see
\cite{Schrijver-SDP}, \cite{Bachoc-Vallentin-kissingnumber},
\cite{Vallentin-symmetry}, \cite{quadruple-distances} and the survey
\cite{Bachoc-ITW}.

\subsection{Organization of the paper}

In Section~\ref{sec 2} we review the classical bounds for Grassmann codes and
the Etzion-Vardy bound for projective codes. In Section~\ref{sec 3} we present
the semidefinite programming method in connection with the theta
number. We show that most of the bounds for Grassmann codes can be
derived from this method. In Section~\ref{sec 4} we reduce the semidefinite
program by the action of the group $GL_n(\F_q)$. In Section~\ref{sec 5} we
present numerical results obtained with this method and we compare
them with the Etzion-Vardy method for $q = 2$ and $n\leq 16$.  Another
distance of interest on projective space, the \emph{injection
  distance}, was introduced in~\cite{network-metrics}. We show how to
modify the Etzion-Vardy bound as well as the semidefinite programming
bound for this.

\section{Elementary bounds for Grassmann and projective
  codes}
\label{sec 2}

\subsection{Bounds for Grassmann codes}

In this section we review the classical bounds for $A_q(n,k,2\delta)$.
We note that the subspace distance takes only even values on the
Grassmann space and that one can restrict to $k\leq n/2$ by the
relation $A_q(n,k, 2\delta) = A_q(n,n-k, 2\delta)$, which follows by
considering orthogonal subspaces.

We recall the definition of the $q$-analog of the binomial coefficient
that counts the number of $k$-dimensional subspaces of a fixed
$n$-dimensional space over $\mathbb{F}_q$, i.e.~the number of elements
of $\mathcal{G}_q(n,k)$.

\begin{Def}
The \emph{$q$-ary binomial coefficient} is defined by
$$
\qbinom{n}{k}_q = \frac{(q^n-1)\dots(q^{n-k+1}-1)}{(q^k-1)\dots (q-1)}.
$$
\end{Def}

\subsubsection{Sphere-packing bound} 

\begin{equation}
\label{sphere-packing-bound}
A_q(n,k,2\delta)\leq \frac{|\mathcal{G}_q(n,k)|}{|B_k(\delta-1)|}=\frac{\qbinom{n}{k}_q}
{\sum_{m=0}^{\lfloor(\delta-1)/2\rfloor}\qbinom{k}{m}_q\qbinom{n-k}{m}_qq^{m^2}}
\end{equation}
It follows from the well-known observation that balls of radius
$\delta-1$ centered at elements of a code $\mathcal{C}\subset
\mathcal{G}_q(n,k)$ with minimal distance $2\delta$ are pairwise
disjoint and have the same cardinality
$\sum_{m=0}^{\lfloor(\delta-1)/2\rfloor}\qbinom{k}{m}_q\qbinom{n-k}{m}_qq^{m^2}$.

\subsubsection{Singleton bound  \cite{Koetter-Kschischang}}

\begin{equation}\label{singleton-bound}
A_q(n,k,2\delta)\leq 
\qbinom{n-\delta+1}{k-\delta+1}_q
\end{equation}
It is obtained by the introduction of a ``puncturing'' operation on
the code.

\subsubsection{Anticode bound  \cite{WXSN}} 

An \emph{anticode} of diameter $e$ is a subset of a metric space whose
pairwise distinct elements are at distance less or equal than $e$.
The general anticode bound (see~\cite{Delsarte-these}) states that, given
a metric space $X$ which is homogeneous under the action of a group
$G$, for every code $\mathcal{C}\subset X$ with minimal distance $d$
and for every anticode ${\mathcal A}$  of diameter $d-1$, we
have
$$
|\mathcal{C}|\leq \frac{|X|}{|\mathcal{A}|}\ .
$$ 
Spheres of given radius $r$ are anticodes of diameter
$2r$. So if we take $\mathcal{A}$ to be a sphere of radius $\delta-1$
in $\mathcal{G}_q(n,k)$ we recover the sphere-packing
bound. Obviously, to obtain the strongest bound, we have to choose the
largest anticodes of given diameter, which in our case are not
spheres. Indeed, the set of all elements of
$\mathcal{G}_q(n,k)$ which contain a fixed $(k-\delta +1)$-dimensional
subspace is an anticode of diameter $2\delta-2$ with
$\qbinom{n-k+\delta-1}{\delta-1}_q$ elements and which is in general
larger than the sphere of radius $\delta-1$. Moreover Frankl and
Wilson proved in \cite{Frankl-Wilson} that  
these anticodes have the largest possible size.
 Taking such $\mathcal{A}$ in the general anticode bound, we
recover \emph{the} (best) anticode bound for $\mathcal{G}_q(n,k)$:
\begin{equation}\label{anticode-bound}
\begin{array}{ll}
 \displaystyle A_q(n,k,2\delta)\leq
 \frac{\qbinom{n}{k}_q}{\qbinom{n-k+\delta-1}{\delta-1}_q}& \displaystyle =
\frac{\qbinom{n}{k-\delta+1}_q}{\qbinom{k}{k-\delta+1}_q}\\
& \displaystyle  = \frac{ (q^n-1)(q^{n-1}-1)\dots (q^{n-k+\delta}-1)}
  {(q^k-1)(q^{k-1}-1)\dots (q^{\delta}-1)}
\end{array}
\end{equation}
It follows from the previous discussion that the anticode bound
improves the sphere-packing bound.  Moreover, the anticode bound is
usually stronger than the Singleton bound, with equality only in the
cases $n = k$ or $\delta = 1$, see~\cite{Xia-Fu}.

\subsubsection{First and second Johnson-type bound \cite{Xia-Fu}}

\begin{equation}\label{first-J-bound}
A_q(n,k,2\delta)\leq \left\lfloor \frac{(q^n-1)(q^k-q^{k-\delta})}{(q^k-1)^2-(q^n-1)(q^{k-\delta}-1)} \right\rfloor
\end{equation}
as long as $(q^k-1)^2-(q^n-1)(q^{k-\delta}-1)>0$, and 
\begin{equation}\label{second-J-bound}
A_q(n,k,2\delta)\leq \left \lfloor \frac{q^n-1}{q^k-1}A_q(n-1,k-1,2\delta)\right \rfloor.
\end{equation}
These bounds were obtained in \cite{Xia-Fu} through the construction of a binary
constant-weight code associated to every constant-dimension code. 
Iterating the latter, one  obtains
\begin{equation}\label{anticode+floor}
A_q(n,k,2\delta)\leq
\left \lfloor \frac{q^n-1}{q^k-1}
\left \lfloor \frac{q^{n-1}-1}{q^{k-1}-1}
\dots
\left \lfloor \frac{q^{n-k+\delta}-1}{q^{\delta}-1}
\right \rfloor\ \dots \right \rfloor\ \right \rfloor.
\end{equation}
If the floors are removed from the right hand side of
\eqref{anticode+floor}, the anticode bound is recovered, so \eqref{anticode+floor} is stronger.
In the particular case of $\delta=k$ and if $n \not \equiv 0\mbox{ mod
}k$, \eqref{anticode+floor} was sharpened in \cite{Etzion-Vardy} 
to
\begin{equation}
\label{A}
A_q(n,k,2k) \leq \left\lfloor \frac{q^n-1}{q^k -1} \right\rfloor -1.
\end{equation}
For $\delta=k$ and if $k$ divides $n$, we have equality in (\ref{anticode+floor}), because of the existence of \emph{spreads} (see \cite{Etzion-Vardy})
$$
A_q(n,k,2k) = \frac{q^n -1}{q^k -1}.
$$
Summing up, the strongest upper bound for constant dimension codes
reviewed so far comes by putting together (\ref{anticode+floor}) and (\ref{A}):

\begin{Thm}\label{Bound grassmann}
If $n-k \not\equiv 0\mbox{ mod }\delta$, then$$A_q(n,k,2\delta)\leq \left\lfloor \frac{q^n-1}{q^k-1} \left\lfloor \dots \left\lfloor \frac{q^{n-k+\delta+1}-1}{q^{\delta+1}-1} \left( \left\lfloor \frac{q^{n-k+\delta}-1}{q^{\delta}-1} \right\rfloor -1 \right) \right\rfloor \dots \right\rfloor \right\rfloor$$
otherwise$$A_q(n,k,2\delta)\leq \left\lfloor \frac{q^n-1}{q^k-1} \left\lfloor \dots \left\lfloor \frac{q^{n-k+\delta+1}-1}{q^{\delta+1}-1} \left\lfloor \frac{q^{n-k+\delta}-1}{q^{\delta}-1} \right\rfloor \right\rfloor \dots \right\rfloor \right\rfloor.$$
\end{Thm}

\subsection{A bound for projective codes}

Here we turn our attention to projective codes whose codewords have
not necessarily the same dimension, and we review the bound obtained
by Etzion and Vardy in~\cite{Etzion-Vardy}. The idea is to split a
code $\mathcal C$ into subcodes ${\mathcal C}_k={\mathcal C}\cap
\mathcal{G}_q(n,k)$ of constant dimension, and then to derive linear
constraints on the cardinality $|{\mathcal C}_k|$, coming from packing
constraints.

Let $B(V,e) := \{U\in {\mathcal P}(\F_q^n) : d_S(U,V)\leq e\}$ denote
the ball with center $V$ and radius $e$.  If $\dim V = i$, we have
$$
|B(V,e)|= \sum_{\ell=0}^e \sum_{j=0}^{\ell} \qbinom{i}{j}_q\qbinom{n-i}{\ell-j}_qq^{j(\ell-j)}.
$$
We define $c(i,k,e):=|B(V,e)\cap \mathcal{G}_q(n,k)|$ for $V$ of
dimension $i$.  It is not difficult to prove that
\begin{equation}\label{c(i,k,e)}
c(i,k,e)=\sum_{j=\lceil \frac{i+k-e}{2} \rceil}^{\min\{k,i\}}\qbinom{i}{j}_q\qbinom{n-i}{k-j}_qq^{(i-j)(k-j)}.
\end{equation}

\begin{Thm}[Linear programming bound for codes in $\mathcal{P}(\mathbb{F}_q^n)$, \cite{Etzion-Vardy}]\label{EV}
\begin{equation*}
 \begin{array}{ccl} A_q(n,2e+1) \leq 
\max \Big\{ \ \displaystyle \sum_{k=0}^nx_k & : & x_k\leq
A_q(n,k,2e+2) \text{ for } k=0,\dots,n,\\
\ &\ & \displaystyle\sum_{i=0}^nc(i,k,e)x_{i}\leq \qbinom{n}{k}_q
\text{ for } k=0,\dots,n\ \Big\}
\end{array} 
\end{equation*}
\end{Thm}
\begin{proof}
For $\mathcal{C} \subset \mathcal{P}(\mathbb{F}_q^n)$ of minimal
distance $2e+1$, and for $k=0,\dots,n$, we introduce  $x_k=|\mathcal{C}\cap
\mathcal{G}_q(n,k)|$. Then $\sum_{k=0}^n x_k=|\mathcal{C}|$ and each
$x_k$ represents the cardinality of a  subcode of $\mathcal{C}$ of
constant dimension $k$, so it is upper bounded by
$A_q(n,k,2e+2)$. Moreover, balls of radius $e$ centered at the
codewords are pairwise disjoint, so the sets $B(V,e)\cap \mathcal{G}_q(n,k)$ for $V\in \mathcal C$ are pairwise disjoint
subsets of $\mathcal{G}_q(n,k)$. So
\begin{equation*}
\sum_{V\in \mathcal C} |B(V,e)\cap \mathcal{G}_q(n,k)|\leq
|\mathcal{G}_q(n,k)|.
\end{equation*}
Because $|B(V,e)\cap \mathcal{G}_q(n,k)| = c(i,k,e)$ if $\dim(V) = i$ we obtain the second constraint
\begin{equation*}
\sum_{i=0}^n c(i,k,e) x_i\leq \qbinom{n}{k}_q.
\end{equation*}
So $|\mathcal C|$ is at most the optimal value of the linear program above.
\end{proof}

\begin{Rem}
Of course, in view of explicit computations, if the exact value of $A_q(n,k,2e+2)$ is not available, it
can be replaced in the linear program of Theorem \ref{EV} by an upper
bound.
\end{Rem}

\section{The semidefinite programming method}
\label{sec 3}

\subsection{Semidefinite programs}

A (real) semidefinite program is an optimization problem of the form:
$$
\sup\left\{\begin{array}{lll} \langle C,Y\rangle & \mbox{:} & Y\succeq 0,\ \langle A_i,Y\rangle=b_i \mbox{ for }
    i=1,\ldots, m\end{array}\right\},
$$
where\begin{itemize}
\item[] $C, A_1, \dots , A_m$ are given real symmetric matrices,
\item[] $b_1,\dots , b_m$ are given real values,
\item[] $Y$ is a real symmetric matrix, which is the optimization variable,
\item[] $\langle A, B\rangle=\operatorname{trace}(AB) $ is the inner
  product between symmetric matrices,
\item[] $Y \succeq 0$ denotes that $Y$ is symmetric and positive semidefinite.
\end{itemize}
This formulation includes linear programming as a special case when
all matrices involved are diagonal matrices. When the input data satisfies some technical assumptions (which are fulfilled for our application) then
there are polynomial time algorithms which determine an approximate optimal value.
We refer to \cite{Todd} and
\cite{SDP} for further details.

\subsection{Lov\'asz' theta number}

In \cite{Lovasz-Shannoncapacity}, Lov\'asz gave an upper bound on the
independence number $\alpha({\mathcal G})$ of a graph $\mathcal{G}=(V,E)$
as the optimal value $\vartheta(\mathcal G)$ of a semidefinite
program:

\begin{Thm}[\cite{Lovasz-Shannoncapacity}]
\begin{equation}\label{Lovasz}
\begin{array}{ccl}
\alpha(\mathcal{G}) \leq \vartheta({\mathcal G}):=\max
\Big\{ 
\sum\limits_{(x,y)\in V^2}F(x,y) & : & F\in {\mathbb R}^{V\times V},\  F \succeq 0,\\[-0.4ex]
\ & \ & \sum\limits_{x\in V} F(x,x)=1,\\[-0.2ex]
\ & \ & F(x,y)=0\text{ if } xy \in E
\Big\}\end{array}
\end{equation}
\end{Thm}

Here we can write $\max$ instead of $\sup$ because one can show using
duality theory of semidefinite programming that the supremum is
attained: the Slater condition \cite[Theorem 3.1]{SDP} is fulfilled.

In the above and all along this paper, we identify a matrix indexed by
a given finite set $V$ with a function defined on $V^2$.  The program given in \eqref{Lovasz} is
one of the many equivalent formulations of Lov\'asz' original
$\vartheta(\mathcal G)$. If the constraint that $F$ only attains
nonnegative values is added, the optimal value gives a sharper bound
for $\alpha(\mathcal G)$. Traditionally this semidefinite program is
denoted by $\vartheta'(\mathcal G)$ \cite{Schrijver-comparaison}.

This method applies to bound the maximal cardinality
$\mathcal{A}(X,d)$ of codes in a metric space $X$ with prescribed
minimal distance $d$.  Indeed $\mathcal{A}(X,d) = \alpha(\mathcal{G})$
where $\mathcal{G}$ is the graph with vertex set $X$ and edges set
$\{xy : 0<d(x,y)<d\}$.  So, we obtain:

\begin{Cor}[The semidefinite programming bound]
\label{SDP-cor}
\begin{equation}\label{SDP-primal}
\begin{array}{lll}
\mathcal{A}(X,d) \leq \max\Big\{ \sum\limits_{(x,y)\in X^2}F(x,y) & : & F\in {\mathbb
  R}^{X\times X},\  F \succeq 0,\ F\geq 0,\\
& & \sum\limits_{x\in X}F(x,x)=1,\\
& & F(x,y)=0\text{ if }0<d(x,y)<d\Big\}
\end{array}
\end{equation}

\begin{equation}\label{SDP-dual}
\begin{array}{lll}
 \phantom{\mathcal{A}(X,d)} = \min \Big\{ \ t/\lambda & : & F\in {\mathbb R}^{X\times X},\   F-\lambda \succeq 0,\\
& & F(x,x) \leq t\ \text{ for all } x\in X,\\
& & F(x,y) \leq 0\ \text{ if } d(x,y) \geq d \Big\}
\end{array}
\end{equation}
\end{Cor}

The second semidefinite program \eqref{SDP-dual} is the dual
of~\eqref{SDP-primal}. Furthermore, by weak duality,
any feasible solution of the semidefinite program in \eqref{SDP-dual}
leads to an upper bound for $\mathcal{A}(X,d)$.

\subsection{Bounds for Grassmann codes}

In two-point homogeneous spa\-ces, the semi\-definite program in
\eqref{SDP-dual} collapses to the linear program of Delsarte, first
introduced in \cite{Delsarte-these} in the framework of association
schemes. This fact was first recognized in the case of the Hamming
space, independently in \cite{McEliece-Lovasz} and
\cite{Schrijver-comparaison}. We refer to \cite{inv-sdp} for a general
discussion on how~\eqref{SDP-primal} and~\eqref{SDP-dual} reduce 
to linear programs in the case of two-point homogeneous spaces.

Grassmann space $\mathcal{G}_q(n,k)$ is two-point
homogeneous for the action of the group $G=GL_n(\F_q)$ and  its associated zonal
polynomials are computed in \cite{Delsarte-Hahn}. They belong to the
family of $q$-Hahn polynomials, which are $q$-analogs of the Hahn
polynomials related to the binary Johnson space.

\begin{Def}\label{q-Hahn} 
  The $q$-Hahn polynomials associated to the parameters $n,s,t$ with
  $0\leq s\leq t\leq n$ are the polynomials $Q_{\ell}(n,s,t; u)$ with
  $0\leq {\ell}\leq \min(s,n-t)$ uniquely determined by the properties:
\begin{enumerate}
\item[(a)] $Q_{\ell}$ has degree ${\ell}$ in the variable $[u]=q^{1-u}\qbinom{u}{1}_q$
\item[(b)] They are orthogonal polynomials for the weights
\begin{equation*}
0\leq i\leq \min(s,n-t) \quad w(n,s,t; i)=\qbinom{s}{i}_q\qbinom{n-s}{t-s+i}_qq^{i(t-s+i)}
\end{equation*}
\item[(c)] $Q_{\ell}(0)=1$.
\end{enumerate}
\end{Def}

To be more precise, in the Grassmann space $\mathcal{G}_q(n,k)$, the
zonal polynomials are associated to the parameters $s=t=k$. The other
parameters will come into play when we analyze the full projective
space in Section~\ref{sec 4}.  The resulting linear programming bound is
explicitly stated in \cite{Delsarte-Hahn}:

\begin{Thm}[Delsarte's linear programming bound
  \cite{Delsarte-Hahn}]
\label{LP}
$$
\begin{array}{ccl}  A_q(n,k,2\delta) \leq
\min \Big\{ 
1+f_1+\dots +f_k & : & f_i\geq 0 \text{ for } i=1,\dots,k,\\
\ &\ & F(u) \leq 0 \text{ for }  u=\delta,\dots,k\Big\},
\end{array} 
$$
where $F(u)=1+\sum_{i=1}^kf_iQ_i(u)$ and $Q_i(u)=Q_i(n,k,k;u)$ as in
Definition~\ref{q-Hahn}.
\end{Thm}

In order to show the power of the semidefinite programming bound, we
will verify that most of the bounds in Section~\ref{sec 2}
for Grassmann codes can be obtained from Corollary \ref{SDP-cor} or
Theorem \ref{LP}.  In each case we construct a suitable feasible
solution of \eqref{SDP-dual}.

\subsubsection{Singleton bound}
We fix an arbitrary subspace $w$ of $\F_q^n$ of dimension $n-\delta+1$. We consider a function 
$\phi:\mathcal{G}_q(n,k)\to \{u\subset w : \dim(u)=k-\delta+1\}$ such that $\phi(x)\subset x$ for all $x$.
Clearly $\dim(x\cap w)\geq k-\delta+1$. In the case of equality, we set $\phi(x)=x\cap w$. If $\dim(x\cap w)> k-\delta+1$,
$\phi(x)$ is chosen arbitrarily among the $(k-\delta+1)$-dimensional subspaces of $x\cap w$. 

We define the function
\begin{align*}
F(x,y)&=
\begin{cases} 1 \text{ if } \phi(x)=\phi(y)\\ 0 \text{ otherwise } \end{cases}\\
&=\sum_{\substack{u\subset w\\ \dim(u)=k-\delta+1}} \mathbf{1}(\phi(x)=u)\mathbf{1}(\phi(y)=u)
\end{align*}
where $\mathbf{1}(\phi(x)=u)$ denotes the characteristic
function of the set $\{x \in \mathcal{G}_q(n,k) : \phi(x) = u\}$.  Then, $F$ is
obviously positive semidefinite, and  $(F,t,\lambda)$ is a
feasible solution of \eqref{SDP-dual} where $t=1$ and 
\begin{align*}
\lambda&= \qbinom{n}{k}_q^{-2}\sum_{(x,y)\in \mathcal{G}_q(n,k)^2} F(x,y)\\
&=\qbinom{n}{k}_q^{-2}\sum_{\substack{u\subset w\\ \dim(u)=k-\delta+1}} \Big(\sum_{x\in \mathcal{G}_q(n,k)}\mathbf{1}(\phi(x)=u)\Big)^2.
\end{align*}
It follows from Cauchy-Schwarz inequality  that $\lambda\geq \qbinom{n-\delta+1}{k-\delta+1}_q$ so  the Singleton bound \eqref{singleton-bound} is recovered 
from \eqref{SDP-dual}.

\subsubsection{Sphere-packing and anticode bounds} 

The sphere-packing bound and the anticode bound in
$\mathcal{G}_q(n,k)$ can also be obtained directly, with 
$$
F(x,y)=\sum_{\dim(z)=k}\mathbf{1}_{B(z,\delta
  -1)}(x)\mathbf{1}_{B(z,\delta -1)}(y),
$$
and
$$
F(x,y)=\sum_{\dim(z)=k-\delta+1}\mathbf{1}(z\subset
x)\mathbf{1}(z\subset y)\ .
$$
In general, the anticode bound $|\mathcal{C}|\leq |X|/|\mathcal{A}|$
can be derived from \eqref{SDP-dual}, using the function
$F(x,y)=\sum_{g\in
  G}\mathbf{1}_{\mathcal{A}}(gx)\mathbf{1}_{\mathcal{A}}(gy)$.

\subsubsection{First Johnson-type bound}

We want to apply Delsarte's linear programming bound of Theorem
\ref{LP} with a function $F$ of degree~$1$, i.e.\
$F(u)=f_0Q_0(u)+f_1Q_1(u)$.  According to \cite{Delsarte-Hahn} the
first $q$-Hahn polynomials are$$Q_0(u)=1\ \ ,\ \
Q_1(u)=\left(1-\frac{(q^n-1)(1-q^{-u})}{(q^k-1)(q^{n-k}-1)}\right).$$
In order to construct a feasible solution of the linear program, we
need $f_0, f_1 \geq 0$ for which $F(u)=f_0+f_1Q_1(u)$ is non-positive
for $u=\delta,\dots,k$. Then $1+f_1/f_0$ will be an upper bound for
$A_q(n,k,2\delta)$. As $Q_1(u)$ is decreasing, the optimal choice of
$(f_0,f_1)$ satisfies $F(\delta)=0$. So $f_1/f_0=-1/Q_1(\delta)$ and
we need $Q_1(\delta)<0$. We obtain \eqref{first-J-bound}:
$$
A_q(n,k,2\delta)\leq
1+\frac{f_1}{f_0}=1-\frac{1}{Q_1(\delta)}=\frac{(q^n-1)(q^{k}-q^{k-\delta})}{(q^k-1)^2-(q^n-1)(q^{k-\delta}-1)}.
$$

\subsubsection{Second Johnson-type bound} 

Here we find an inequality for the optimal value $B_q(n,k,2\delta)$ of
the semidefinite program~\eqref{SDP-dual} in the case
$X=\mathcal{G}_q(n,k)$ (with the subspace distance) which
resembles~\eqref{second-J-bound}:
$$
B_q(n,k,2\delta) \leq \frac{q^n-1}{q^k-1}B_q(n-1,k-1,2\delta).
$$
Let $(F,t,\lambda)$ be an optimal solution for the program
\eqref{SDP-dual} in $\mathcal{G}_q(n-1,k-1)$ relative to the minimal
distance $2\delta$, i.e. $F$ satisfies the conditions: $F \succeq
\lambda$, $F(x,x) \leq t$, $F(x,y)\leq 0\mbox{ if }d(x,y)\geq
2\delta$, and $t/\lambda = B_q(n-1,k-1,2\delta)$. We consider the
function $G$ on $\mathcal{G}_q(n,k) \times \mathcal{G}_q(n,k)$ given
by
$$
G(x,y)=\sum_{\dim(D)=1}\mathbf{1}(D \subset x)\mathbf{1}(D \subset y)F(x
\cap H_D, y \cap H_D),
$$
where, for every one-dimensional space $D$, $H_D$ is an arbitrary
hyperplane such that $D \oplus H_D = \mathbb{F}_q^n$.  It can be
verified that the triple $(G,t',\lambda')$ is a feasible solution of the program
\eqref{SDP-dual} in $\mathcal{G}_q(n,k)$ for the minimal distance
$2\delta$, where $t'=t\qbinom{k}{1}_q$ and $\lambda'=\lambda\qbinom{k}{1}_q^2/\qbinom{n}{1}_q$, thus leading to the upper bound
$$
B_q(n,k,2\delta) \leq \frac{t'}{\lambda'} = \frac{t}{\lambda} \frac{q^n-1}{q^k-1} =
\frac{q^n-1}{q^k-1}B_q(n-1,k-1,2\delta).
$$

\begin{Rem} In \cite{Etzion-Vardy} another Johnson-type bound is given:
$$
A_q(n,k,2\delta) \leq \frac{q^n-1}{q^{n-k}-1}A_q(n-1,k,2\delta),
$$
which follows easily from the second Johnson-type bound combined with
the equality $A_q(n,k,2\delta)=A_q(n,n-k,2\delta)$. Similarly to
above, an analogous inequality holds for the semidefinite programming
bound~$B_q(n,k,2\delta)$.
\end{Rem}

\section{Semidefinite programming bounds for projective codes}
\label{sec 4}

In this section we perform a symmetry reduction of the semidefinite
programs \eqref{SDP-primal} and \eqref{SDP-dual} in the case of
projective space, under the action of the group $G =
GL_n(\mathbb{F}_q)$. We follow the general method described in
\cite{inv-sdp}. The key point is that these semidefinite programs are
left invariant under the action of $G$ so the set of feasible
solutions can be restricted to $G$-invariant functions~$F$. The main
work is to compute an explicit description of the $G$-invariant
positive semidefinite functions on the projective space.

\subsection{$G$-invariant positive semidefinite functions on projective spaces}

In order to compute these functions, we use the decomposition of
the space of real-valued functions under the action of $G$. We take
the following notations:
\[
X=\mathcal{P}(\mathbb{F}_q^n), \quad X_k=\mathcal{G}_q(n,k), \quad \mathbb{R}^X=\{f:X\rightarrow
\mathbb{R}\}.
\]
The space $\R^X$ is endowed with the inner product $(,)$ defined by:
\begin{equation*}
(f,g)=\frac{1}{|X|}\sum_{x\in X} f(x)g(x).
\end{equation*}
For $k=0,\dots,n$, an element of $\R^{X_k}:=\{f:X_k\to \R\}$ is identified with the
element of $\R^X$ that takes the same value on $X_k$ and the value $0$
outside of $X_k$. In this way, we see the spaces $\R^{X_k}$ as
pairwise orthogonal subspaces of $\R^X$.

Delsarte~\cite{Delsarte-Hahn} showed that the irreducible
decomposition of the $\mathbb{R}^{X_k}$ under the action of $G$ is
given by the \emph{harmonic subspaces} $H_{k,i}$:
\begin{equation}
\label{Grassmann-decomposition}
\mathbb{R}^{X_k}=H_{0,k} \oplus H_{1,k} \oplus \dots \oplus
H_{\min\{k,n-k\},k}
\end{equation}
Here, $H_{k,k}$ is the kernel of the \emph{differentiation operator}
$$
\begin{array}{clcl}
\delta_k: & \mathbb{R}^{X_k} & \longrightarrow & \mathbb{R}^{X_{k-1}}\\
\ & f & \longrightarrow & \left[\ x \rightarrow \sum\{f(y): \dim(y)=k, x \subset y\}\ \right]
\end{array}
$$
and $H_{k,i}$ is the image of $H_{k,k}$ under the \emph{valuation operator}
$$
\begin{array}{clcl}
\psi_{ki}: &\mathbb{R}^{X_k} & \longrightarrow & \mathbb{R}^{X_{i}}\\
\ & f & \longrightarrow & \left[\ x \rightarrow \sum\{f(y): \dim(y)=k, y \subset x\}\ \right]
\end{array}
$$
for $k\leq i \leq n-k$.
Because $\delta_k$ is surjective, we have $h_k:=\dim(H_{k,k}) =
\qbinom{n}{k}_q - \qbinom{n}{k-1}_q.$ Moreover, $\psi_{ki}$ commutes
with the action of $G$, so $H_{k,i}$ is isomorphic to $H_{k,k}$. 
Putting together the spaces  $\mathbb{R}^{X_k}$ one gets the global
picture:
\begin{center}
\small
$$
\begin{array}{ccccccccccccccccc}
\mathbb{R}^X & = & \mathbb{R}^{X_0} & \oplus & \mathbb{R}^{X_1} & \oplus
& \cdots & \oplus & \mathbb{R}^{X_{\left\lfloor \frac{n}{2} \right
    \rfloor}} & \oplus & \cdots & \oplus & \mathbb{R}^{X_{n-1}}
&\oplus & \mathbb{R}^{X_n}\\
& & & & & & & & & & & & & & \\
\mathcal{I}_0 & = & H_{0,0} & \oplus &  H_{0,1} & \oplus & \cdots &
\oplus & H_{0,\left\lfloor \frac{n}{2} \right\rfloor}  & \oplus & \cdots & \oplus & H_{0, (n-1)} & \oplus & H_{0,n}\\
\mathcal{I}_1 & = & & &  H_{1,1} & \oplus & \cdots & \oplus &
H_{1,\left\lfloor \frac{n}{2} \right\rfloor} & \oplus &  \cdots & \oplus & H_{1, (n-1)} & & &\\
\mathcal{I}_2 & = & & &  & & \cdots & \oplus &
H_{2,\left\lfloor \frac{n}{2} \right\rfloor} & \oplus &  \cdots &  &
& & &\\
\vdots & = & & &  & & \ddots & & & &   &  &
& & &\\
\vdots & = & & &  & &  & \ddots & &  &  &  &
& & &\\
\mathcal{I}_{\left\lfloor \frac{n}{2} \right \rfloor} & = & & & & & & &
H_{\left\lfloor \frac{n}{2} \right\rfloor,\left\lfloor \frac{n}{2}
  \right\rfloor} &\oplus & \cdots& & & & &\\
\end{array}
$$
\end{center}
Here, the columns give the irreducible decomposition
(\ref{Grassmann-decomposition}) of the spaces $\mathbb{R}^{X_k}$.  The
irreducible components which lie in the same row are all isomorphic,
and together they form the \emph{isotypic
  components}
$$
\mathcal{I}_m:=H_{m,m}\oplus H_{m,m+1} \oplus \dots
\oplus H_{m,n-m} \simeq H_{m,m}^{n-2m+1}.
$$ 
Starting from this decomposition, one builds the \emph{zonal matrices}
$E_k(x,y)$ \cite[Section 3.3]{inv-sdp} in the following way. We take
an isotypic component $\mathcal{I}_k$ and we fix an orthonormal basis
$(e_{kk1}, \dots, e_{kkh_k})$ of $H_{k,k}$. Let
$e_{ksi}:=\psi_{ks}(e_{kki})$. It follows from  \cite[Theorem
3]{Delsarte-Hahn} that $(e_{ks1}, \dots,
e_{ksh_k})$ is an orthogonal basis of $H_{k,s}$ and that 
\begin{equation}\label{inner product}
(e_{ksi},e_{ksi})=\qbinom{n-2k}{s-k}_q q^{k(s-k)}.
\end{equation}
Then we define $E_k(x,y) \in
\mathbb{R}^{(n-2k+1)\times (n-2k+1)}$ entrywise by
\begin{equation}\label{defEkst}
E_{kst}(x,y)=\sum_{i=1}^{h_k}e_{ksi}(x)e_{kti}(y).
\end{equation}
We note that \cite[Theorem
3.3]{inv-sdp} requires orthonormal basis in every subspace, while  in
the definition \eqref{defEkst} of $E_{kst}$ we do not normalize the
vectors $e_{ksi}$. Because the norms \eqref{inner
  product} do not depend on $i$, but only on $k,s$, the matrix $(E_k'(x,y))_{s,t}$ associated to the normalized basis
is obtained from $(E_k(x,y))_{s,t}$ by left and  right multiplication by a diagonal matrix. 
So the  characterization of the $G$-invariant positive semidefinite
functions given in  \cite[Theorem
3.3]{inv-sdp} holds aswell with \eqref{defEkst}:
\begin{Thm}
\label{blockdiag}
$F \in \mathbb{R}^{X\times X}$ is positive semidefinite and $G$-invariant if and only if it can be written as
\begin{equation}
\label{Bochner-blockdiag}
F(x,y)=\sum_{k=0}^{\lfloor n/2\rfloor} \langle F_k, E_k(x,y)\rangle
\end{equation}
where  $F_k\in {\mathbb R}^{(n-2k+1)\times (n-2k+1)}$ and
$F_0,\dots,F_{\lfloor n/2\rfloor}$ are positive semidefinite.
\end{Thm}

Now we compute the $E_k$'s explicitly. They are zonal matrices: in other
words, for all $k\leq s, t \leq n-k$, for all $g \in G,
E_{kst}(x,y)=E_{kst}(gx,gy)$. This means that $E_{kst}$ is a
function of the variables which parametrize the orbits of $G$ on
$X \times X$. It is easy to see that the orbit of the pair
$(x,y)$ is characterized by the triple $(\dim(x),\dim(y),\dim(x\cap y))$.

The next theorem gives an explicit expression of $E_{kst}$, in terms
of the polynomials $Q_k$ of Definition \ref{q-Hahn}. 

\begin{Thm}
\label{Ekst}
If $k\leq s\leq t\leq n-k$, $\dim(x)=s$, $\dim(y)=t$,
\begin{equation*} 
E_{kst}(x,y)= |X| h_k\frac{\qbinom{t-k}{s-k}_q\qbinom{n-2k}{t-k}_q}{\qbinom{n}{t}_q\qbinom{t}{s}_q}q^{k(t-k)}Q_k(n,s,t; s-\dim(x\cap y))
\end{equation*}
If $\dim(x)\neq s$ or $\dim(y)\neq t$, $E_{kst}(x,y)=0$.
\end{Thm}

We note that the weights involved in the orthogonality relations of
the polynomials $Q_k$ have a combinatorial meaning:

\begin{Lemma}[\cite{Dunkl}]
\label{weights}
Given $x\in X_s$, the number of elements $y\in X_t$ such that $\dim(x\cap y)=s-i$ is equal to $w(n,s,t;i)$.
\end{Lemma}

\begin{proof}[Proof of Theorem \ref{Ekst}]
  By construction, $E_{kst}(x,y) \neq 0$ only if $\dim(x)= s$ and
  $\dim(y)= t$, so in this case $E_{kst}$ is a function of
  $(s-\dim(x\cap y))$.  Accordingly, for $k\leq s \leq t\leq n-k $, we
  introduce $P_{kst}$ such that
  $E_{kst}(x,y)=P_{kst}(s-\dim(x\cap y))$. Now we want to relate
  $P_{kst}$ to the $q$-Hahn polynomials. We start with two lemmas:
  one obtains the orthogonality relations satisfied by $P_{kst}$ and
  the other computes $P_{kst}(0)$.

\begin{Lemma} 
With the above notations,
\begin{equation}\label{e3}
P_{kst}(0)=|X| h_k
\frac{\qbinom{t-k}{s-k}_q\qbinom{n-2k}{t-k}_q}{\qbinom{n}{t}_q\qbinom{t}{s}_q}q^{k(t-k)}.
\end{equation}
\end{Lemma}

\begin{proof} 
  We have $P_{kst}(0)=E_{kst}(x,y)$ for all $x,y$ with
  $\dim(x)=s$, $\dim(y)=t$, $x\subset y$. Hence,
\begin{align*}
P_{kst}(0) &= \frac{1}{\qbinom{n}{t}_q\qbinom{t}{s}_q}
\sum_{\substack{\dim(x)=s\\ \dim(y)=t\\x\subset y}} E_{kst}(x,y)\\
&= \frac{1}{\qbinom{n}{t}_q\qbinom{t}{s}_q}
\sum_{\substack{\dim(x)=s\\ \dim(y)=t\\x\subset y}} 
\sum_{i=1}^{h_k} e_{ksi}(x)e_{kti}(y)\\
\end{align*}
\begin{align*}
P_{kst}(0) &= \frac{1}{\qbinom{n}{t}_q\qbinom{t}{s}_q} 
\sum_{i=1}^{h_k}\sum_{\dim(y)=t} \Big(\sum_{\substack{\dim(x)=s \\x\subset y}}  e_{ksi}(x)\Big)e_{kti}(y)\\
&= \frac{1}{\qbinom{n}{t}_q\qbinom{t}{s}_q} 
\sum_{i=1}^{h_k}\sum_{\dim(y)=t} \psi_{s,t}(e_{ksi})(y) e_{kti}(y).\\
\end{align*}
With the relation $\psi_{st}\circ\psi_{ks}=\qbinom{t-k}{s-k}_q\psi_{kt}$,
\begin{equation*}
\psi_{st}(e_{ksi})=\psi_{st}\circ
\psi_{ks}(e_{kki})=\qbinom{t-k}{s-k}_q\psi_{kt}(e_{kki})
=\qbinom{t-k}{s-k}_qe_{kti},
\end{equation*}
and we obtain 
\begin{align*}
P_{kst}(0) &= 
\frac{1}{\qbinom{n}{t}_q\qbinom{t}{s}_q} 
\sum_{i=1}^{h_k}\sum_{\dim(y)=t} \qbinom{t-k}{s-k}_qe_{kti}(y) e_{kti}(y)\\
&= 
\frac{\qbinom{t-k}{s-k}_q}{\qbinom{n}{t}_q\qbinom{t}{s}_q}  
\sum_{i=1}^{h_k} |X|(e_{kti}, e_{kti} )
=
|X|h_k\frac{\qbinom{t-k}{s-k}_q\qbinom{n-2k}{t-k}_q}{\qbinom{n}{t}_q\qbinom{t}{s}_q}q^{k(t-k)}.
\end{align*}
\end{proof}

\begin{Lemma} 
With the above notation,
\begin{equation}\label{e4}
\sum_{i=0}^s w(n,s,t;i)P_{kst}(i)P_{{\ell}st}(i)= \delta_{k,{\ell}}|X|^2h_k
\frac{\qbinom{n-2k}{s-k}_q
\qbinom{n-2k}{t-k}_qq^{k(s+t-2k)}}{\qbinom{n}{s}_q}.
\end{equation}
\end{Lemma}

\begin{proof}
  We compute $\Sigma:=\sum_{y\in X} E_{kst}(x,y) E_{{\ell},s',t'}(y,z)$.
\begin{align*}
\Sigma &= 
\sum_{y\in X} \sum_{i=1}^{h_k}\sum_{j=1}^{h_{\ell}}
e_{ksi}(x)e_{kti}(y)e_{{\ell}s'j}(y)e_{{\ell}t'j}(z)\\
&= \sum_{i=1}^{h_k}\sum_{j=1}^{h_{\ell}}
e_{ksi}(x)e_{{\ell}t'j}(z)\Big(\sum_{y\in
  X}e_{kti}(y)e_{{\ell}s'j}(y)\Big)\\
&= \sum_{i=1}^{h_k}\sum_{j=1}^{h_{\ell}}
e_{ksi}(x)e_{{\ell}t'j}(z)|X|(e_{kti},e_{{\ell}s'j})\\
&= \sum_{i=1}^{h_k}\sum_{j=1}^{h_{\ell}}
e_{ksi}(x)e_{{\ell}t'j}(z)|X|\delta_{k{\ell}}\delta_{ts'}\delta_{ij}\qbinom{n-2k}{t-k}_qq^{k(t-k)}\\
&=
\delta_{k{\ell}}\delta_{ts'}|X|\qbinom{n-2k}{t-k}_qq^{k(t-k)}\sum_{i=1}^{h_k}
e_{ksi}(x)e_{kt'i}(z)\\
&= \delta_{k{\ell}}\delta_{ts'}|X|\qbinom{n-2k}{t-k}_qq^{k(t-k)}E_{kst'}(x,z).
\end{align*}
We obtain, with $t=s'$, $t'=s$, $x=z\in X_s$, and taking 
$E_{{\ell}ts}(y,x)=E_{{\ell}st}(x,y)$ into account, 
\begin{equation*}
\sum_{y\in X_t} E_{kst}(x,y)E_{{\ell}st}(x,y)
=\delta_{k{\ell}}|X|\qbinom{n-2k}{t-k}_qq^{k(t-k)} E_{kss}(x,x).
\end{equation*}
The above identity becomes in terms of $P_{kst}$
\begin{equation*}
\sum_{y\in X_t} P_{kst}(s-\dim(x\cap y))P_{{\ell}st}(s-\dim(x\cap y))
=\delta_{k{\ell}}|X|\qbinom{n-2k}{t-k}_qq^{k(t-k)} P_{kss}(0).
\end{equation*}
Now we obtain~\eqref{e4} by \eqref{e3} and Lemma \ref{weights}.
\end{proof}

We showed that the functions $P_{kst}$ satisfy the same orthogonality
relations as the $q$-Hahn polynomials. So we are done if 
$P_{kst}$ is a polynomial of degree at most $k$ in the
variable $[u]=[\dim(x\cap y)]$. This property is proved in the case
$s=t$ in \cite[Theorem 5]{Delsarte-Hahn} and extends to $s\leq t$ with
a similar line of reasoning. The multiplicative factor between $P_{kst}(u)$ and
$Q_k(n,s,t;u)$ is then given by $P_{kst}(0)$ and the proof of
Theorem \ref{Ekst} is completed.
\end{proof}

\subsection{Symmetry reduction of the semidefinite program \eqref{SDP-primal} for
  projective codes}

Clearly, \eqref{SDP-primal} is $G$-invariant: this means that for
every feasible solution $F$ and for every $g \in G$, also $gF$ is
feasible with the same objective value.  Hence, we can average every
feasible solution over $G$. In particular, the optimal value
of~\eqref{SDP-primal} is attained by a function $F$ which is
$G$-invariant and so we can restrict the optimization variable
in~\eqref{SDP-primal} to be a $G$-invariant function.

A function $F(x,y) \in \mathbb{R}^{X\times X}$ is $G$-invariant if it
depends only on $\dim(x)$, $\dim(y)$, and $\dim(x\cap y)$. So
we introduce $\tilde{F}$, such that $F(x,y)=\tilde{F}(s,t,i)$ for
$x,y \in X$ with $\dim(x)=s,\dim(y)=t,\dim(x\cap y)=i$. Let
\begin{equation*}
N_{sti} :=|\{(x,y)\in X\times X : \dim(x)=s,\dim(y)=t,\dim(x\cap
y)=i\}|
\end{equation*}
and
\begin{align}\label{Omega}
\Omega(d):=\{(s,t,i) :\  &0\leq s,t\leq n,\  0\leq i\leq \min(s,t), \ s+t\leq n+i,\\\nonumber
&\text{either } s=t=i\text{ or }s+t-2i \geq d\}.
\end{align}
Then, \eqref{SDP-primal} becomes:
\begin{equation*}
\begin{array}{llll}
A_q(n,d) \leq \max\Big\{\displaystyle \sum_{s,t,i}N_{sti}\tilde{F}(s,t,i) & : & \tilde{F}\in \mathbb{R}^{[n]^3},\  \tilde{F} \succeq 0,\ \tilde{F}\geq 0,\\[-0.6ex]
&& \displaystyle \sum_{s=0}^nN_{sss}\tilde{F}(s,s,s)=1,\\
&& \displaystyle \tilde{F}(s,t,i) =0 \text{ if } (s,t,i)\notin
\Omega(d)\ \Big\},
\end{array}
\end{equation*}
where $\tilde{F} \succeq 0$ means that the corresponding $F$ is
positive semidefinite.

Then, we introduce the variables
$x_{sti}:=N_{sti}\tilde{F}(s,t,i)$. It is straightforward to rewrite
the program in terms of these variables, except for the condition
$\tilde{F}\succeq 0$.  From Theorem~\ref{blockdiag}, this is equivalent to the
semidefinite conditions $F_k \succeq 0$, where the matrices $F_k$ are
given by the scalar product of $F$ and $E_k$:
\begin{align*}
(F_k)_{st}=\frac{1}{|X|^2h_k\qbinom{n-2k}{s-k}_q  q^{k(s-k)}\qbinom{n-2k}{t-k}_q q^{k(t-k)}}
\sum_{(x,y)\in  X^2}F(x,y)E_{kst}(x,y)\\
=\frac{1}{|X|^2h_k\qbinom{n-2k}{s-k}_q  q^{k(s-k)}\qbinom{n-2k}{t-k}_q
  q^{k(t-k)}}
\sum_{u,v,i}x_{uvi}\tilde{E}_{kst}(u,v,i)
\end{align*}
We can substitute the value of $\tilde{E}_{kst}(u,v,i)$ using
Theorem~\ref{Ekst}; in particular it vanishes  when $(u,v)\neq(s,t)$,
and, when $(u,v)=(s,t)$ and $s\leq t$:
\begin{equation}\label{Fk}
(F_k)_{st}=\frac{1}{|X|}\sum_{i}x_{sti}\frac{\qbinom{t-k}{s-k}_q}{\qbinom{n}{t}_q\qbinom{t}{s}_q\qbinom{n-2k}{s-k}_q}q^{-k(s-k)}Q_k(n,s,t; s-i).\end{equation}

\begin{Thm}\label{SDP-final}
\begin{equation*}
\begin{array}{llll}
A_q(n,d) \leq \max\Big\{ \displaystyle \sum_{(s,t,i)\in
  \Omega(d)}x_{sti} & : & (x_{sti})_{(s,t,i)\in \Omega(d)},\ x_{sti}\geq 0,\\[-0.6ex]
&& \displaystyle\sum_{s=0}^nx_{sss}=1,\\
&& F_k \succeq 0\ \text{ for all } k=0,\dots, \lfloor n/2 \rfloor\Big\},
\end{array}
\end{equation*}
where $\Omega(d)$ is defined in \eqref{Omega} and the matrices $F_k$ are given in (\ref{Fk}).
\end{Thm}

\begin{Rem}\label{rem} A projective code $\mathcal C$ with minimal distance
  $d$ provides a feasible solution of the above program, given by:
\begin{equation*}
x_{sti}=\frac{1}{|\mathcal C|} |\{(x,y)\in \mathcal{C} :
\dim(x)=s,\dim(y)=t,\dim(x\cap y)=i\}.
\end{equation*}
In particular, we have
\begin{equation*}
\sum_{t,i} x_{sti}= |\mathcal{C}\cap \mathcal{G}_q(n,s)|,
\end{equation*}
so, we can add the valid inequality
\begin{equation*}
\sum_{t,i} x_{sti} \leq A_q(n,s,2\lceil d/2\rceil)
\end{equation*}
to the semidefinite program of Theorem \ref{SDP-final} in order to tighten it.

Following the same line of reasoning, we could also add the linear inequalities
\begin{equation*}
\sum_{s=0}^n c(s,k,e) \big(\sum_{t,i} x_{sti}\big) \leq \qbinom{n}{k}_q, \ k=0,\dots,n
\end{equation*}
where $e=\lfloor (d-1)/2\rfloor$, to the semidefinite program of Theorem \ref{SDP-final}, so that the resulting semidefinite program
contains all the constraints of the  linear program of Theorem \ref{EV}. It turns out that  
this semidefinite program behaves numerically badly, and that, when it can be computed, its optimal value 
is equal to the minimum of the optimal values of the initial semidefinite program and of the linear program.
\end{Rem}

\section{Numerical results}
\label{sec 5}

In this section we report the numerical results obtained for the binary case $q =
2$. Table 1 contains upper bounds for $A_2(n,d)$ for the subspace
distance $d_S$ while Table 2 contains upper bounds for
$A^{inj}_2(n,d)$ for the injection distance $d_i$ recently introduced
in~\cite{network-metrics}.

\subsection{Subspace distance}

The first column of Table 1 displays the upper bound obtained from
Etzion-Vardy's linear program, Theorem~\ref{EV}. Observing that the variables $x_k$ in this program represent integers,
its optimal value as an integer program gives an upper bound for $A_q(n,2e+1)$ that may improve on the optimal value 
of the linear program in real variables. However, we observed a difference with 
the optimal value of the linear program in real variables of at most $1$. In Table 1, we display the bound obtained with the optimal value of the linear program in real variables, and indicate with a superscript $*$ the cases when the integer program 
gives a better bound (of one less).

The second column contains the upper bound from the semidefinite
program of Theorem \ref{SDP-final}, strengthened by the inequalities
(see Remark \ref{rem}):
\begin{equation*}
\sum_{t,i} x_{sti} \leq A_2(n,s,2\lceil d/2\rceil)\mbox{\ \ for all }s=0,\dots, n.
\end{equation*}

In both programs, $A_2(n,k,2\delta)$ was replaced by the upper bound from Theorem \ref{Bound grassmann}.

\begin{table}[htbp]
\begin{center}
\begin{tabular}{r|r|r}
\textbf{parameter} & \textbf{E-V LP} & \textbf{SDP}\\
\hline
$A_2(4,3)$ & 6 &  6\\
$A_2(5,3)$ & 20 &  20\\
$A_2(6,3)$ & *124 &  124\\
$A_2(7,3)$ & 832 &  776\\
$A_2(7,5)$ & 36 &  35\\
$A_2(8,3)$ & 9365 &  9268\\
$A_2(8,5)$ & 361 &  360\\
$A_2(9,3)$ & *114387 & 107419\\
$A_2(9,5)$ & *2531 &  2485\\
$A_2(10,3)$ & *2543747 & 2532929\\
$A_2(10,5)$ & *49451 & 49394 \\
$A_2(10,7)$ & *1224 & 1223\\
$A_2(11,5)$ & 693240 & 660285\\
$A_2(11,7)$ & 9120 & 8990\\
$A_2(12,7)$ & 323475 & 323374\\
$A_2(12,9)$ & *4488 & 4487\\
$A_2(13,7)$ & 4781932 & 4691980\\
$A_2(13,9)$ & *34591 & 34306\\
$A_2(14,9)$ & 2334298 & 2334086\\
$A_2(14,11)$ & *17160 & 17159\\
$A_2(15,11)$ & *134687 & 134095\\
$A_2(16,13)$ & *67080 & 67079\\
\end{tabular}
\end{center}
\caption{Bounds for the subspace distance}
\end{table}

\subsection{Additional inequalities}

Etzion and Vardy~\cite{Etzion-Vardy} found additional
valid inequalities for their linear program in the special 
case of $n=5$ and $d=3$. 
With this, they could  improve their bound to the exact value
$A_2(5,3)=18$. In this section we establish analogous inequalities for
other parameters $(n,d)$. 

\begin{Thm}\label{Add} Let $\mathcal{C} \subset \mathcal{P}(\mathbb{F}_q^n)$,  of
  minimal subspace distance $d$,  and let $D_k:=|\mathcal{C}\cap
  \mathcal{G}_q(n,k)|$. Then,  if
 \begin{equation*}
d+2\left\lceil d/2 \right\rceil +2 < 2n < 2d+2\left\lceil d/2
\right\rceil +2,
\end{equation*} we have:
\begin{itemize}
\item $D_{2n-d-\left\lceil d/2 \right\rceil -1} \leq 1$;
\item if $D_{2n-d-\left\lceil d/2 \right\rceil -1} = 1$ then $$D_{\left\lceil d/2 \right\rceil} \leq \frac{q^n-q^{2n-d-\left\lceil d/2 \right\rceil -1}}{q^{\left\lceil d/2 \right\rceil}-q^{n-d-1}}.$$
\end{itemize}
\end{Thm}

\begin{proof}
It is clear that $D_i \leq 1$ for  $0 \leq i < \left\lceil d/2 \right\rceil$.
Moreover, for all $x,y\in \mathcal{C}\cap\mathcal{G}(n,\left\lceil d/2
\right\rceil)$, $x\neq y$,  $\dim(x\cap y)=0$.
We want to show that $D_{2n-d-\left\lceil d/2 \right\rceil -1} \leq 1$.
Indeed assume by contradiction $x \neq y \in
\mathcal{C}\cap\mathcal{G}(n,2n-d-\left\lceil d/2 \right\rceil -1)$,
we have
$$ \left\{ \begin{array}{l}  4n-2d-2\left\lceil d/2 \right\rceil -2 \leq n+\dim(x\cap y) \\
d \leq 4n-2d-2\left\lceil d/2 \right\rceil-2-2\dim(x\cap y)
\end{array} \right.$$
leading to
$$ \left\{ \begin{array}{ll} 2\dim(x\cap y) \geq 6n-4d-4\left\lceil d/2 \right\rceil-4 & (*)\\
2\dim(x\cap y) \leq 4n-3d-2\left\lceil d/2 \right\rceil-2 & (**)
\end{array} \right.$$
To obtain a contradiction, we must have $(*) > (**)$ which is
equivalent to the hypothesis  $2n>d+2\left\lceil d/2 \right\rceil +2$.

With a similar reasoning, we prove that, for all $x \in \mathcal{C}\cap\mathcal{G}(n,\left\lceil d/2 \right\rceil)$ and $w \in \mathcal{C}\cap\mathcal{G}(n,2n-d-\left\lceil d/2 \right\rceil -1)$, $\dim (x\cap w)=n-d-1$.
Indeed, 
$$\left\{\begin{array}{l} 2n-d-1 \leq n + \dim (x\cap w) \\
d \leq 2n-d-1-2\dim (x\cap w) 
\end{array}
\right.$$
so
$$\left\{\begin{array}{l} \dim (x\cap w) \geq n-d-1\\
\dim (x\cap w) \leq n-d-1/2
\end{array}
\right.$$
which yields the result.

Now we assume $D_{2n-d-\left\lceil d/2 \right\rceil -1} = 1$. Let  $w
\in \mathcal{C}\cap\mathcal{G}(n,2n-d-\left\lceil d/2 \right\rceil
-1)$. Let $\mathcal{U}$ denote the union of the subspaces $x$
belonging to $\mathcal{C}\cap\mathcal{G}(n,\left\lceil d/2
\right\rceil)$. 
We have $|\mathcal{U}|=1 + D_{\left\lceil d/2
  \right\rceil}(q^{\left\lceil d/2 \right\rceil}-1)$ and 
$|\mathcal{U}\cap w |=1 + D_{\left\lceil d/2 \right\rceil}(q^{n-d-1}-1)$.
On the other hand,  $|\mathcal{U}\backslash (\mathcal{U}\cap w))| \leq
|\mathbb{F}_q^n\backslash w|$, leading to
$$D_{\left\lceil d/2 \right\rceil}(q^{\left\lceil d/2 \right\rceil}-q^{n-d-1}) \leq q^n - q^{2n-d-\left\lceil d/2 \right\rceil -1}\ .$$
\end{proof}

In several cases, adding these inequalities led to a lower optimal
value, however we found that only in one case other than
$(n,d)=(5,3)$, the final result, after rounding down to an integer, is
improved. It is the case $(n,d)=(7,5)$, where $D_3 \leq 17$ and, by
Theorem \ref{Add}, if $D_5=1$ then $D_3\leq 16$. So we can add
$D_3+D_5\leq 17$ and $D_2+D_4\leq 17$, leading to: $A_2(7,5) \leq 34$.
This bound can be obtained with both the linear program of Theorem
\ref{EV} and the semidefinite program of Theorem \ref{SDP-final}.

\subsection{Injection distance}

Recently, a new metric has been considered in the framework of
projective codes, the \emph{injection} metric, introduced in
\cite{network-metrics}. The \emph{injection distance} between two
subspaces $U,V \in \mathcal{P}(\mathbb{F}_q^n)$ is defined
by
$$
d_i(U,V)=\max\{\dim(U),\dim(V)\}-\dim(U\cap V).
$$ 
When restricted to the Grassmann space, i.e.\ when $U,V$ have the same dimension, the
new distance coincides with the subspace distance (up to
multiplication by $2$). In general we have the relation
(see \cite{network-metrics})
$$
d_i(U,V)=\frac{1}{2}d_S(U,V) +
\frac{1}{2}|\dim(U)-\dim(V)|,
$$
where $d_S$ denotes the subspace distance.

It is straightforward to modify the programs in order to produce
bounds for codes on this new metric space
$(\mathcal{P}(\mathbb{F}_q^n), d_i)$. Let
$$
A_q^{inj}(n,d)=\max\{|\mathcal{C}|\ :\ \mathcal{C}\subset
\mathcal{P}(\mathbb{F}_q^n),\ d_i(\mathcal{C})\geq d\}.
$$
For constant dimension codes, we have $A_q^{inj}(n,k,d)=A_q(n,k,2d).$

To modify the linear program of Etzion and Vardy for this new
distance, we need to write down packing-constraints.
The cardinality of balls in $\mathcal{P}(\mathbb{F}_q^n)$ for the
injection distance can be found in~\cite{inj-codes}.  Let
$B^{inj}(V,e)$ be the ball with center $V$ and radius $e$.  If
$\dim(V)=i$, we have
\begin{align*}
|B^{inj}(V,e)|= &\sum_{r=0}^e q^{r^2}\qbinom{i}{r}_q\qbinom{n-i}{r}_q\\
&+\sum_{r=0}^e\sum_{\alpha=1}^r q^{r(r-\alpha)}\left(\qbinom{i}{r}_q\qbinom{n-i}{r-\alpha}_q +
  \qbinom{i}{r-\alpha}_q\qbinom{n-i}{r}_q\right).
\end{align*}
We define $c^{inj}(i,k,e):=|B^{inj}(V,e)\cap \mathcal{G}_q(n,k)|$
where $\dim(V)=i$. We set $\alpha:=|i-k|$.
$$
c^{inj}(i,k,e)=\left\{\begin{array}{ll}
\sum_{r=0}^eq^{r(r-\alpha)}\qbinom{i}{r}_q\qbinom{n-i}{r-\alpha}_q & \mbox{ if }i \geq k\\
\sum_{r=0}^eq^{r(r-\alpha)}\qbinom{i}{r-\alpha}_q\qbinom{n-i}{r}_q & \mbox{ if }i \leq k\\
\end{array}\right.$$

\begin{Thm}[Linear programming bound for codes in $\mathcal{P}(\mathbb{F}_q^n)$ with injection distance]
$$\begin{array}{ccl} 
A_q^{inj}(n,d) \leq \max \Big\{ \sum_{k=0}^nx_k & : & x_k\leq A_q^{inj}(n,k,d)\ \forall\ k=0,\dots,n\\
\ &\ & \sum_{i=0}^nc^{inj}(i,k,e)x_{i}\leq \qbinom{n}{k}_q\ \forall\
k=0,\dots,n \Big\}
\end{array} 
$$
\end{Thm}

For the semidefinite programming bound, we only need to change the
definition of $\Omega(d)$ by
\begin{align}\label{Omega inj}
\Omega^{inj}(d):=\{(s,t,i) :\  &0\leq s,t\leq n,\  i\leq \min(s,t), \ s+t\leq n+i,\\\nonumber
&\text{either } s=t=i\text{ or }\max(s,t)-i \geq d\}.
\end{align}

\begin{Thm}
\label{SDPinj-final}
\begin{equation*}
\begin{array}{llll}
A_q^{inj}(n,d) \leq \max\Big\{ \displaystyle \sum_{(s,t,i)\in
  \Omega^{inj}(d)}x_{sti} & : & (x_{sti})_{(s,t,i)\in \Omega^{inj}(d)},\ x_{sti}\geq 0,\\[-0.4cm]
&& \displaystyle\sum_{s=0}^nx_{sss}=1,\\
&& F_k \succeq 0\ \text{ for all } k=0,\dots, \lfloor n/2 \rfloor\Big\}
\end{array}
\end{equation*}
where $\Omega^{inj}(d)$ is defined in \eqref{Omega inj} and the matrices $F_k$ are given in (\ref{Fk}).
\end{Thm}

Table 2 displays the numerical computations we obtained from the two
programs.

\begin{table}[htbp]
\begin{center}
\begin{tabular}{r|r|r}
\textbf{parameter} & \textbf{E-V LP}  & \textbf{SDP}\\
\hline
$A^{inj}_2(7,3)$ & 37 &  37 \\
$A^{inj}_2(8,3)$ & 362 &  364\\
$A^{inj}_2(9,3)$ & 2533 &  2536\\
$A^{inj}_2(10,3)$ & 49586 & 49588\\
$A^{inj}_2(10,4)$ & 1229 &  1228\\
$A^{inj}_2(11,4)$ & 9124 & 9126\\
$A^{inj}_2(12,4)$ & 323778 & 323780\\
$A^{inj}_2(12,5)$ & 4492 & 4492\\
$A^{inj}_2(13,5)$ & 34596 & 34600\\
$A^{inj}_2(14,6)$ & 17167 & 17164\\
$A^{inj}_2(15,6)$ & 134694 & 134698 \\
$A^{inj}_2(16,7)$ & 67087 & 67084\\
\end{tabular}
\end{center}

\caption{Bounds for the injection distance}
\end{table}

\begin{Rem} We observe that the bound obtained for $A_2^{inj}(n,2e+1)$  is
  most of the time slightly larger than the one obtained for
  $A_2(n,4e+1)$. In \cite{inj-codes}, the authors noticed that
  their constructions led to codes that are slightly better for the
  injection distance that for the subspace distance. So both
  experimental observations indicate that $A_2(n,4e+1)$ is larger than
  $A_2^{inj}(n,2e+1)$.
\end{Rem}

The computational part of this research would not have been possible
without the use of free software: We computed the values of the linear
programs with Avis' lrs 4.2 available from
http://cgm.cs.mcgill.ca/\~{}avis/C/lrs.html. The values of the
semidefinite programs we computed with SDPA or SDPT3, available from
the NEOS website http://www.neos-server.org/neos/.

\section*{Acknowledgements} We would like to thank the
first referee for  valuable comments and suggestions.

\medskip

{\it E-mail address:} christine.bachoc@math.u-bordeaux1.fr\\
\indent{\it E-mail address:} alberto.passuello@math.u-bordeaux1.fr \\
\indent{\it E-mail address:} frank.vallentin@uni-koeln.de

\end{document}